\documentclass[12pt]{article}
\usepackage{amsfonts,amssymb,amsbsy,amsmath,amsthm,enumerate}
\usepackage[dvips,bookmarks,colorlinks=true,linkcolor=blue,urlcolor=red]{hyperref}

\newtheorem{theorem}{Theorem}

\newtheorem{lemma}[theorem]{Lemma}

\theoremstyle{definition}

\newcommand{\bel}{\begin{equation} \label}
\newcommand{\ee}{\end{equation}}

\newcommand{\rd}{{\mathbb R}^{2}}
\newcommand{\re}{{\mathbb R}}
\newcommand{\R}{{\mathbb R}}
\newcommand{\N}{{\mathbb N}}

\def\beq{\begin{equation}}
\def\eeq{\end{equation}}
\newcommand{\bea}{\begin{eqnarray}}
\newcommand{\eea}{\end{eqnarray}}
\newcommand{\beas}{\begin{eqnarray*}}
\newcommand{\eeas}{\end{eqnarray*}}

{

\begin{document}
  \begin{center} {\LARGE \bf The Fate of the Landau Levels under
      Perturbations of Constant Sign}

    \medskip

    \today
  \end{center}

  \medskip

  \begin{center}
    {\sc Fr{\'e}d{\'e}ric Klopp, Georgi  Raikov}\\
  \end{center}

\begin{abstract}
  We show that the Landau levels cease to be eigenvalues if we perturb
  the 2D Schr{\"o}dinger operator with constant magnetic field, by bounded
  electric potentials of fixed sign. We also show that, if the
  perturbation is not of fixed sign, then any Landau level may be an
  eigenvalue  of the
  perturbed problem.\\
\end{abstract}

{\bf AMS 2000 Mathematics Subject Classification:} 35J10, 81Q10,
35P20\\

{\bf Keywords:}
Landau Hamiltonians,  splitting of Landau levels \\

\section{Introduction. Main results}
\setcounter{equation}{0}
 In this note we consider the
Landau Hamiltonian $H_0$, i.e. the 2D Schr{\"o}dinger operator
with constant magnetic field. It is well-known that the spectrum
of $H_0$ consists of an arithmetic progression of eigenvalues
called {\em Landau levels} of infinite multiplicity. In Theorem
\ref{th1} we show that under  perturbations by fairly general
electric potentials of constant sign, the Landau levels cease to
be eigenvalues of the perturbed operator. Moreover, in Theorem
\ref{th2} we show that for each fixed Landau level there exist
non-constant-sign electric potentials such that the Landau levels
is still an eigenvalue of  infinite multiplicity of the perturbed operator.\\
Let
$$
H_0 : = \left(-i\frac{\partial}{\partial x} +\frac{by}{2}\right)^2
+ \left(-i\frac{\partial}{\partial y} - \frac{bx}{2}\right)^2 - b
$$
be the Landau Hamiltonian shifted by the value $b>0$ of the
constant magnetic field. The operator $H_0$ is self-adjoint in
$L^2(\rd)$, and essentially self-adjoint on $C^{\infty}_0(\rd)$.
 Note that
$C_0^{\infty}(\rd)\setminus\{0\}$ is a form core for the operator
$H_0$. It is well-known (see \cite{f, l, ahs}) that the spectrum
$\sigma(H_0)$ of the operator $H_0$ consists of the so-called
Landau levels $2bq$, $q \in {\mathbb N} : = \{0,1,2\ldots\}$,
which are eigenvalues of $H_0$ of infinite multiplicity.\\
 Let $V
\in L^{\infty}(\rd;\re)$. We will suppose that
    \bel{1}
    c \chi({\bf x})
    \leq V({\bf x}), \quad {\bf x} = (x,y) \in \rd,
    \ee
    where $c>0$ is
a constant and $\chi$ is the characteristic function of a disk of
radius $r>0$ in $\rd$, and \bel{2} \|V\|_{L^{\infty}(\rd)} < 2b.
\ee Set $H_{\pm} : = H_0 \pm V$. The main result of the note is
the following
\begin{theorem} \label{th1} Fix $q \in {\mathbb N}$. \\
  (i) Assume that $V \in L^{\infty}(\rd;\re)$ satisfies \eqref{1}; if
  $q \geq 1$, suppose in addition that \eqref{2} holds true. Then we
  have \bel{3} {\rm Ker}\,(H_+ - 2bq) = \{0\}.  \ee (ii) Assume that
  $V$ satisfies \eqref{1} and \eqref{2}. Then we have \bel{3a} {\rm
    Ker}\,(H_- - 2bq) = \{0\}.  \ee
\end{theorem}
\noindent The proof of Theorem \ref{th1} can be found in Section
2.

To the authors' best knowledge the fate of the Landau levels under
perturbations of the described class had never been addressed in the
mathematical literature. However, the asymptotic distribution of the
discrete spectrum near the Landau levels of various perturbations of
the Landau Hamiltonian and its generalizations has been considered by
numerous authors (see \cite{r0, i, rw, mr, fp, pr, rt, rs, pe}); in
particular, it was shown in \cite{rw} that for any $V$ which satisfies
\eqref{1}, and is relatively compact with respect to $H_0$, and for
any Landau level there exists an infinite sequence of discrete
eigenvalues of $H_{\pm}$ which accumulates to this Landau level. Such
results are related to the problem treated here: indeed, the existence
of such an infinite sequence is a necessary condition that the Landau
level not be an infinitely
degenerate eigenvalue of $H_{\pm}$.\\

\par The fact that $V$ has a fixed sign plays a crucial role in
our result, as shows the following
\begin{theorem}
  \label{th2}
  Fix $q\geq0$. Then, there exists a bounded compactly supported
  non-constant-sign potential $V$ such that
  $\|V\|_{L^{\infty}(\rd)}<b$ and
  \begin{equation}
    \label{eq:4}
    \dim\,{\rm Ker}\,(H_0+V - 2bq) = \infty.
  \end{equation}
\end{theorem}
\noindent The proof of Theorem~\ref{th2} is contained in Section
3. Its strategy is to consider radially symmetric potentials $V$,
and, applying a decomposition into a Fourier series with respect
to the angular variable, to represent the operator $H_0 + V$ as an
infinite sum of ordinary differential operators involving only the
radial variable. Such a representation of $H_0 + V$ is well known,
and has been used in different contexts of the spectral theory of
the perturbed Landau Hamiltonian (see e.g. \cite{ahs,ms}). To
prove Theorem~\ref{th2}, the basic consequence  is that, for a
compactly supported, radially symmetric potential $V$, the first
derivative with respect to the coupling constant $\lambda$ at
$\lambda=0$ of the eigenvalues of $H_0+\lambda V$ close to the
$q$-th Landau level is determined by $V$ near the external rim of
its support. Thus, writing $V=V_t$ as an infinite sum of
concentric potentials depending on different coupling constant
$t=(t_l)_{l\geq1}\in\ell^\infty(\N^*)$, one can construct an
analytic mapping from a neighborhood of $0$ in
$\in\ell^\infty(\N^*)$ to a subset of the eigenvalues of $H_0+V_t$
near the $q$-th Landau level, the Jacobian of which we control for
$t=0$.

The potential exhibited in Theorem~\ref{th2} can be chosen arbitrarily
small. Following the same idea, one can also construct compactly
supported potentials such that any of the Landau levels be of finite
non trivial multiplicity or non compactly supported, bounded
potentials such that~\eqref{eq:4} be satisfied for any $q\in\N$.

\section{Proof of Theorem \ref{th1}} \setcounter{equation}{0}
Denote by $\Pi_q$, $q \in {\mathbb N}$, the orthogonal projection
onto ${\rm Ker}(H_0 - 2bq)$. Set
$$
\Pi_q^+ : = \sum_{j=q}^{\infty} \Pi_j, \quad \Pi_q^- : = I -
\Pi_q^+, \quad q \in {\mathbb N}.
$$
In order to prove Theorem \ref{th1}, we need a technical result
concerning some Toeplitz-type operators of the form $\Pi_q V
\Pi_q$.
\begin{lemma} \label{l1} Let $V \in L^{\infty}(\rd;\re)$ satisfy
  \eqref{1}. Fix $q \in {\mathbb N}$. Then
  \bel{6} \langle \Pi_q V \Pi_q u, u \rangle =
  0, \quad u \in L^2(\rd), \ee
where $\langle
  \cdot , \cdot \rangle$ denotes the scalar product in $L^2(\rd)$,
    implies \bel{7} \Pi_q u = 0.  \ee
\end{lemma}
\begin{proof}  By \eqref{1} and
  \eqref{6},
  \bel{4} 0 \leq c\langle \Pi_q \chi \Pi_q u,
  u\rangle \leq \langle \Pi_q V \Pi_q u, u\rangle = 0,
  \ee
  i. e.
  \bel{5} \langle \Pi_q \chi \Pi_q u, u\rangle = 0. \ee   Denote by $T
  : = \Pi_q \chi \Pi_q$ the operator self-adjoint in the Hilbert space
  $\Pi_q L^2(\rd)$.  The operator $T$ is positive and compact, and its
  eigenvalues can be calculated explicitly (see \cite[Eq.
  (3.32)]{rw}). This explicit calculation implies that ${\rm Ker}\,T =
  \{0\}$. Therefore,  \eqref{7} follows from \eqref{5}.
\end{proof}
\begin{proof}[Proof of Theorem \ref{th1}] First, we prove \eqref{3} in
  the case $q=0$.  Assume that there exists $ u \in D(H_+) = D(H_0)$
  such that $H_+u = 0$. Hence, \bel{15} \langle H_0 u, u \rangle +
  \langle Vu,u \rangle = 0.  \ee The two terms at the l.h.s. of
  \eqref{15} are non-negative, and therefore they both should be equal
  to zero. Since $\langle H_0 u, u\rangle = 0$, we have \bel{16} u =
  \Pi_0 u.  \ee Therefore, $\langle V u, u\rangle = \langle \Pi_0 V
  \Pi_0 u, u\rangle = 0$. By Lemma
  \ref{l1}, we have $\Pi_0 u = 0$, and by \eqref{16} we conclude that
  $u = 0$. \\
  Next, we prove \eqref{3} in the case $q \geq 1$.  Assume that there
  exists $ u \in D(H_0)$ such that \bel{9} H_+ u = 2bq u.  \ee Set
  $u_+ : = \Pi_q^+ u$, $u_- : = u - u_+$; evidently, $u_{\pm} \in
  D(H_0)$.  Since $H_0$ commutes with the projections $\Pi_q^{\pm}$,
  \eqref{9} implies \bel{10} H_0 u_+ - 2bq u_+ + \Pi_q^+ V \Pi_q^+ u_+
  + \Pi_q^+ V \Pi_q^- u_- = 0, \ee \bel{11} H_0 u_- - 2bq u_- +
  \Pi_q^- V \Pi_q^- u_- + \Pi_q^- V \Pi_q^+ u_+ = 0.  \ee Now note
  that the operator $H_0 + \Pi_q^- V \Pi_q^- - 2bq$ is boundedly
  invertible in $\Pi_q^- L^2(\rd)$, and its inverse is a negative
  operator.  Moreover, by \eqref{11} we have \bel{14} u_- = -
  \left(H_0 + \Pi_q^- V \Pi_q^- - 2bq\right)^{-1}\Pi_q^- V \Pi_q^+ u_+
  , \ee which inserted into \eqref{10} implies
  $$
  H_0 u_+ - 2bq u_+ + \Pi_q^+ V \Pi_q^+ u_+ - \Pi_q^+ V \Pi_q^-
  \left(H_0 + \Pi_q^- V \Pi_q^- - 2bq\right)^{-1}\Pi_q^- V \Pi_q^+ u_+
  = 0,
  $$
  and hence,
  $$
  \langle (H_0 - 2bq) u_+, u_+\rangle + \langle \Pi_q^+ V \Pi_q^+ u_+,
  u_+\rangle $$ \bel{12} - \langle \Pi_q^+ V \Pi_q^- \left(H_0 +
    \Pi_q^- V \Pi_q^- - 2bq\right)^{-1}\Pi_q^- V \Pi_q^+ u_+,
  u_+\rangle = 0.  \ee The three terms on the l.h.s. of \eqref{12} are
  non-negative, and hence they all should be equal to zero.  Since
  $u_+ = \Pi_q^+ u_+$, the equality $\langle (H_0 - 2bq) u_+,
  u_+\rangle = 0$ implies \bel{13} u_+ = \Pi_q u_+.  \ee Therefore,
  $\langle \Pi_q^+ V \Pi_q^+ u_+, u_+\rangle = \langle \Pi_q V \Pi_q
  u_+, u_+\rangle$, and $\langle \Pi_q^+ V \Pi_q^+ u_+, u_+\rangle = 0$
  is equivalent to $\langle \Pi_q V \Pi_q u_+, u_+\rangle = 0$.  Now
  by Lemma \ref{l1} we have $\Pi_q u_+ = 0$, by \eqref{13} we have
  $u_+ = 0$, and by \eqref{14} we have $u_- =
  0$. Therefore, $u=0$.\\
  Finally, we sketch the proof of \eqref{3a} which is quite similar to the one of \eqref{3}.
Let $w \in
  D(H_0)$, $H_- w = 2bq w$.  Set $w_+ : =
  \Pi_{q+1}^+ w$, $w_- : = w - w_+$.
   The operator $H_0 -
  \Pi_{q+1}^+ V \Pi_{q+1}^+ - 2bq$ is boundedly invertible in
  $\Pi_{q+1}^+ L^2(\rd)$, its inverse is a positive operator, and
  by analogy with \eqref{14} we get
     $
     w_+ =
\left(H_0 - \Pi_{q+1}^+ V \Pi_{q+1}^+ - 2bq\right)^{-1}\Pi_{q+1}^+
V
  \Pi_{q+1}^- w_-$.
  Further, similarly to \eqref{12}, we find
  that
  $$
  \langle (H_0 - 2bq) w_-, w_-\rangle - \langle \Pi_{q+1}^+ V
  \Pi_{q+1}^- w_-, w_-\rangle $$ $$ - \langle \Pi_{q+1}^- V
  \Pi_{q+1}^+ \left(H_0 - \Pi_{q+1}^+ V \Pi_{q+1}^+ -
    2bq\right)^{-1}\Pi_{q+1}^+ V \Pi_{q+1}^- w_-, w_-\rangle = 0.
    $$
  The three terms on the l.h.s.  are non-positive, and
  hence they should vanish. As in the proof of \eqref{3}, we easily conclude
  that $w_- = 0$, and hence $w=0$.
\end{proof}

\section{Proof of Theorem \ref{th2}}
\setcounter{equation}{0} Define the operators
$$
H_0^{(m)} : = - \frac{1}{\varrho} \frac{d}{d \varrho} \varrho
\frac{d}{d \varrho} + \left(\frac{m}{\varrho} - b \varrho\right)^2
- b, \quad m \in {\mathbb Z},
$$
 self-adjoint in $L^2(\re_+; \varrho d\varrho)$, as the
Friedrichs' extensions of the operators defined on
$C_0^{\infty}(\re_+)$ with $\re_+ : = (0,\infty)$.  Then, the
operator $H_0$ is unitarily equivalent to the orthogonal sum
$\oplus_{m \in
  {\mathbb Z}} H_0^{(m)}$ under the passage to polar coordinates
$(\varrho,\phi)$ in $\rd$, and a subsequent decomposition into a
Fourier series with respect to the angular variable $\phi$. For
any $m \in {\mathbb Z}$, we have
$$
\sigma(H_0^{(m)}) = \bigcup_{q=m_-}^{\infty}\{2bq\}
$$
where, as usual, $m_- : = \max \{0,-m\}$ (see e.g. \cite{ahs}).
In contrast to the 2D Landau Hamiltonian $H_0$ however, we have $
{\rm
  dim \, Ker} (H_0^{(m)} - 2bq) = 1$ for all $q
\geq m_-$, $m \in {\mathbb Z}$. \\
Further, assume that $V \in L^{\infty}(\rd;\re)$ and $V$ is
radially symmetric i.e.
$$
V(x,y) = v\left(\sqrt{x^2 + y^2}\right), \quad (x,y) \in \rd.
$$
Then, the operator $H_0 + V$ is unitarily equivalent to the
orthogonal
sum $\oplus_{m \in {\mathbb Z}} (H_0^{(m)} +v)$.\\
Thus,
\bel{cu3} {\rm dim \, Ker} (H_0 + V - 2bq) = \sum_{m \in
  {\mathbb Z}} {\rm dim \, Ker} (H_0^{(m)} + v - 2bq), \quad q \in
{\mathbb N}. \ee
If $\|V\|_{L^{\infty}(\rd)}=\|v\|_{L^{\infty}(\R_+)} < b$, for all
$m \in \N$, the $q$-th eigenvalue of $H_0^{(m)} + v$ that we
denote by $E_q(v;m)$, stays in the interval $2bq+]-b,b[$; in
particular, it stays simple. So, as a consequence of regular
perturbation theory, see e.g.~\cite{Ka:80,Re-Si}, the eigenvalues
$(E_q(v;m))_{q\geq0}$ are real analytic functions of the potential
$v$. Moreover, one computes
\begin{equation}
  \label{eq:2}
  \frac{\partial }{\partial t}E_q(tv;m)|_{t=0}=
  \int_{\R_+}v(\rho)\varphi_{q,m}(\varrho)^2 \varrho d\varrho
\end{equation}
where
\begin{equation*}
\varphi_{q,m}(\varrho) : = \sqrt{\frac{q!}{\pi
    (q+m)!}\left(\frac{b}{2}\right)^{m+1}} \varrho^m
L_q^{(m)}\left(b\varrho^2/2\right) e^{-b \varrho^2/4}, \quad
\varrho \in \re_+, \quad q \in {\mathbb N},
\end{equation*}
are the normalized eigenfunctions of the operator $H_0^{(m)}$, $m
\in {\mathbb N}$, and
$$
L_q^{(m)}(s) : = \sum_{l=0}^{q} \frac{(q+m)!}{(m+l)! (q-l)!}
\frac{(-s)^l}{l!}, \quad s \in \re,
$$
are the generalized Laguerre polynomials.\\
Pick $t\in ]-b/2,b/2[^{\N^*}$ and consider the potential
\begin{equation}
  \label{eq:1}
  v_t(\rho)= -
  \sum_{j\in\N^*}t_{2j-1}\bold{1}_{[x^-_{2j-1},x^+_{2j-1}]}(\rho) +
  \sum_{j\in\N^*}t_{2j}\bold{1}_{[x^-_{2j},x^+_{2j}]}(\rho), \quad
  \rho \in \re_+,
\end{equation}
where $x^-_j : =e^{-\alpha_j/2}$, $x^+_j := e^{-\beta_j/2}$, and
\begin{equation}
  \label{eq:3}
  \alpha_{2j-1}:=2^{-N(j-1/2)^2+1},\ \beta_{2j-1}:=2^{-Nj^2+1},\
  \alpha_{2j}:=2^{-N (j-1/2)^2},\ \beta_{2j}:=2^{-N j^2}.
\end{equation}
We will choose the large integer $N$ later on.\\
As, for $j\geq1$, one has
\begin{equation*}
N(j-1)^2<N(j-1/2)^2-1<N(j-1/2)^2<Nj^2-1<Nj^2<N(j+1/2)^2-1,
\end{equation*}
we note that, for $N$ sufficiently large, one has:
\begin{itemize}
\item
  $\|v_t\|_{L^{\infty}(\R_+)}<b$ for $t\in ]-b/2,b/2[^{\N^*}$;
\item $v_t$ vanishes identically if and only if the vector
$(t_j)_j$
  vanishes identically.
\end{itemize}
For $j\geq1$, define $m_j=2^{N j^2}-1$ and consider the mapping
\begin{equation*}
  \mathcal{E}:\ t\in ]-b/2,b/2[^{\N^*}\mapsto
  (\mathcal{E}_{2j-1}(t),\mathcal{E}_{2j}(t))_{j \geq 1}=
  (t_{2j}+t_{2j-1},\tilde E_q(v_t;m_j))_{j \geq 1}\in]-r,r[^{\N^*}
\end{equation*}
where
\begin{equation*}
  \tilde E_q(v_t;m_j)=\frac{2\pi q!}{C_j}\frac{m_j (m_j!)^2}{(q+m_j)!}
  \left(\frac{2}{b}\right)^{m_j + 1}\left(E_q(v_t;m_j)-2bq\right)
\end{equation*}
The constants $(C_j)_j$ are going to be chosen later on.\\
The mapping is real analytic and we can compute its Jacobi matrix
at $t=0$. First, bearing in mind \eqref{eq:2}, \eqref{eq:1}, and
\eqref{eq:3}, we easily find that
\begin{gather*}
\begin{split}
  \partial_{t_{2j}}\mathcal{E}_{2l}(0)&=C_j^{-1}(e^{-m_l\beta_{2j}}(1+o(1))
  -e^{-m_l\alpha_{2j}}(1+o(1))) \\&
  =  \begin{cases}
      1\text{ if }j=l,\\
      O\left(e^{-2^{N|j-l|}}\right)\text{ if }l>j,\\
      O\left(2^{-N|j-l|}\right)\text{ if }l<j,
    \end{cases}
    \end{split}
  \\
  \begin{split}
    \partial_{t_{2j+1}}\mathcal{E}_{2l}(0)
    &=-C_j^{-1}(e^{-m_l\beta_{2j+1}}(1+o(1))
    -e^{-m_l\alpha_{2j+1}}(1+o(1)))\\
    &=    \begin{cases}
      -e^{-2} + O(e^{-2^{Nj}}) \text{ if }j=l,\\
      O\left(e^{-2^{N|j-l|}}\right)\text{ if }l>j,\\
      O\left(2^{-N|j-l|}\right)\text{ if }l<j,
    \end{cases}
  \end{split}
\end{gather*}
when one chooses $C_j$ properly. In this formula, $o(1)$ refers to
the behavior of the function when $N\to+\infty$ uniformly in
$l,j$. Moreover, obviously,
$$
\partial_{t_{2j-1}}\mathcal{E}_{2l -
1}(0) = \partial_{t_{2j}}\mathcal{E}_{2l-1}(0) = \delta_{jl}.
$$
\\
Hence,  the Jacobi matrix of the mapping $\mathcal{E}(t)$ at $t=0$
can be written as $J+E$ where $J$ is a block diagonal matrix made
of the blocks $\displaystyle \begin{pmatrix} 1 & 1 \\ -e^{-2} & 1
\end{pmatrix}$ and the error matrix $E$ is a bounded operator from
$l^\infty(\N^*)$ to itself with a norm bounded by $C 2^{-N}$. So
for $N$ large enough this Jacobi matrix is invertible and, using
the analytic inverse mapping theorem, we see that there exists a
real analytic diffeomorphism $\varphi$ on a ball of
$l^\infty(\N^*)$ centered at $0$, such that
\begin{equation*}
  \mathcal{E}\circ\varphi(u)=
  (u_{2j}+u_{2j-1},u_{2j}-e^{-2}u_{2j-1})_{j\geq 1}\in]-r,r[^{\N^*},
\end{equation*}
and $\varphi(0) = 0$.  To construct the potential $v_t$ having the
Landau level $2bq$ as an eigenvalue with infinite multiplicity, it
suffices to take $t=\varphi(u)$ with $u_{2j}= e^{-2}u_{2j-1} \neq 0$
for infinitely many indices $j \in \N^*$. This completes
the proof of Theorem~\ref{th2}.\\

{\bf Acknowledgements}. The authors were partially supported by
the Chilean Scientific Foundation {\em Fondecyt} under Grants
7080135 and 1050716.\\
G. Raikov acknowledges also the partial support of {\em N{\'u}cleo
Cient{\'\i}fico ICM} P07-027-F ``{\em Mathematical Theory of
Quantum and Classical Magnetic Systems"}.

{\sc Fr{\'e}d{\'e}ric Klopp}\\
D{\'e}partement de  math{\'e}matiques\quad\quad et\quad\quad Institut
Universitaire de France\\
Universit{\'e}  de  Paris Nord\\
Avenue  J.Baptiste  Cl{\'e}ment\\
93430 Villetaneuse, France\\
E-mail: klopp@math.univ-paris13.fr\\

{\sc G. Raikov}\\
Facultad de Matem{\'a}ticas\\
Pontificia Universidad Cat{\'o}lica de Chile\\
Av. Vicu{\~n}a Mackenna 4860\\ Santiago de Chile\\
E-mail: graikov@mat.puc.cl\\

\end{document}